\documentclass[12pt,letterpaper]{article}
\usepackage[compact]{titlesec}
\titlespacing{\section}{0pt}{2ex plus .3ex minus .2ex}{1.5ex plus .2ex}  % Adjust spacing before and after section headings
\usepackage{caption}
\usepackage{enumitem} % Use \usepackage{enumerate} if you wanted indented enumeration or \usepackage{enumitem} to avoid indentation in enumeration
\usepackage[font+=footnotesize]{subcaption}
\usepackage{amsmath,epsfig,amsfonts,amssymb,graphicx,bm,amsbsy,amssymb,url,tabularx,latexsym}
\usepackage{graphicx}
	
\newcommand{\xspace}[0]{{\cal X}}

\newcommand{\Mspace}[0]{{\cal M}}

\newcommand{\moreone}[0]{[1,\infty)}

\newcommand{\indicator}[1]{1\left(#1 \right)}

\newtheorem{theorem}{Theorem}
\newtheorem{lemma}{Lemma}
\newtheorem{corollary}{Corollary}
\newenvironment{proof}{\noindent{\em Proof:}}{\hfill $\square$}
\newenvironment{proofnumbered}{\noindent{\em Proof}}{\hfill $\square$}

\title{A Theoretical Framework for Online Information Search}
\author{Rohit Negi  \\
	Carnegie Mellon University  \\
	}
\date{}

\begin{document}

\maketitle

\begin{abstract}
A significant part of human activity today consists of searching for a piece of information online, utilizing knowledge repositories.
This endeavor may be time-consuming if the individual searching for the information is unfamiliar with the subject matter of that
information. However,  experts can aid individuals find relevant information by searching online. This paper describes a theoretical framework
to model the dynamic process by which requests for information come to a system of experts, who then answer the requests by searching for those pieces of information.
\end{abstract}

%---------------------------------------------------------------------------
\section{Introduction}
\label{sec:intro}

The Internet today has been transformed from a network providing connectivity, to a massive repository of human (and machine) knowledge, with 
information relevant to nearly every aspect of human life stored in some corner. Search engines allow keyword-based search of this knowledge, and while
natural language queries are increasingly useful, searching for complex information requires human thinking (augmented with the capabilities of
search engines) to obtain useful search results. While there are canonical `big problems' in different fields that require specialized experts,
a large part of human life deals with a vast number of  small problems, each affecting a different individual in its own unique manner.
These problems require the individual to search the Internet for ideas relevant to solving that problem, an activity that may receive mixed results,
depending on the expertise of that individual. 
But given pervasive online connectivity, there are potentially a large number of `experts' available online that an individual can consult, who can 
contribute their knowledge to problems related to their expertise \cite{wsj}. 

Given the growing importance of such online requests for information, this paper envisions a large number of requests for information being made, but also a large number of potential experts
 available to answer those requests. Since the requests must be responded to in a timely manner, we propose a dynamic framework,
where requests arrive stochastically, are handled by expert(s) who search for relevant information, and depart when the expert provides
a response. Preliminary results on scheduling requests, and on the resulting capacity of the system are presented.

%---------------------------------------------------------------------
\section{Theoretical Framework}
\label{sec:framework}

The problem setting in the paper assumes that requests for information come into a social network stochastically. Each request
is handled by an expert (or experts), which searches for information to answer that request, and succeeds in providing
information answering that request after a random amount of time, based on the complexity of the search. 
This requires describing a quantitative model for information search
and also describing a model for scheduling these requests, so that experts can answer them. 

%---------------------------------------------------------------------
\begin{enumerate}[leftmargin=*]
\item{\bf Model of Information Search}

Time is assumed to be discretized finely, so that it is measured as $t=1,2,3,\ldots$ time slots.
Let $\Mspace$ be a large set of information facts.
A {\em topic} $x \subset \Mspace$ is a large subset of facts - examples being `Windows 10 debugging' or `Seventeenth century poetry'. The set of topics $\xspace$ is assumed to be large but finite to avoid technical clutter.
An {\em expert} is a {\em research time} function $T:\xspace \rightarrow
\moreone$, where 
$T(x)\geq 1$ is the mean time that expert takes to answer a request concerning topic $x$; this average time is assumed to be known to the expert. This time is required because the expert will
typically need to search for information relevant to the specific request before being able to answer it. We assume that the time to answer a specific request is a geometrically distributed random variable
(with mean value $T(x)$). A typical request may be `Why does my Windows 10 laptop become hot and shut down?',
which concerns the topic `Windows 10 debugging'.  For conciseness, we will simply
call a request concerning topic $x$ as {\em request $x$}. 
%Two experts $T_1,T_2$ (or just experts $1,2$ when
%the context is clear)
%are said to have $\varepsilon-$similar expertise if $(1+\varepsilon)^{-1} \leq \frac{T_1(x)}{T_2(x)}\leq (1+\varepsilon), \forall x \in \xspace$.

%We generalize the expert model as follows. An expert  is allowed to work simultaneously on multiple requests  $x_1,x_2,\ldots,x_k$, needing average time $T(x_1,x_2,\ldots,x_k)$ to arrive at
%all the answers. Also, each request $x$ can be worked on by several experts $T_1,T_2,\ldots,T_l$, and answered in average time 
%$T_{1,2,\ldots,l}(x)$, a process we call as {\em collaboration}. 

\item{\bf Model of Dynamic Scheduling}

It is assumed that there is a social network of $n$ experts, represented
 as a graph $G=(V,E)$, where the vertices $V$ represent experts and the edges $E$ represented coordination opportunity between pairs of experts. By {\em coordination}, we mean that
a scheduler (described below) can assign a request in expert $i$'s queue to expert $j$, as long as $(i,j) \in E$ in the graph.
For example, if experts exclusively use a Knowledge market (or an Internet Q\&A Forum) like {\tt Quora}  \cite{quora}, they can all coordinate with each other, and so $G$ is a complete graph.
On the other hand, if a social network like {\tt Twitter} or  {\tt Facebook} is used, the graph may have a complex structure,
precluding arbitrary coordination. This paper only considers a complete graph linking the experts.

We adopt a dynamic stochastic model of information searching. In each time slot $t$, at each expert $i \in V$, each request $x \in \xspace$ may newly arrive with probability $\lambda p_i(x)$, and so,  we need a multi-class queuing model. Denote as $a_{x,i}(t) =0,1$ the non-arrival or arrival of request $x$ at expert $i$, respectively. Its arrival is
independent of arrival of requests in other topics, arrivals at other experts, and arrivals in other time slots. Here, $0 < \lambda <1$ and $p_i(x)>0$ is a probability mass function (p.m.f.) over topics $x$ (so, $\sum_{x \in \xspace} p_i(x)=1$).
$\lambda$ can be interpreted as the {\em request load} on the network, while $p_i(x)$ causes requests for certain  topics to appear more frequently.  Due to independence, we allow multiple different topics $x$ to arrive at any expert, and also multiple experts to see requests from
the same topic $x$. Each expert $i$ puts request  $x$ into its own {\em virtual queue} and increases the length
$Q_{x,i}(t)$ of that topic's queue  by one request  (all requests will actually be written in random access memory, so the virtual queue is a book-keeping artifact). 
In practice, the requests may be given to the expert by users she knows in her social circle, or may be selected by the expert from a knowledge
market like {\tt Quora}.

A scheduler
then assigns different requests to different experts, subject to the social network graph, allowing the experts to {\em coordinate} in handling the requests. 
Since this paper assumes a complete graph model, the scheduler can assign any request to any expert. 
%A scheduler is called {\em centralized} if its uses
%explicit knowledge of the experts' individual research time $T_i(x)$ functions, otherwise it is called {\em distributed}. In the latter case,
%the scheduler may interactively solve an optimization problem with its experts, but it never needs to explicitly know their  $T_i(x)$. This case
%is practically important, because human experts may not be able to provide an explicit function (since $|\xspace|$ is large), but
%they may intuitively know $T(x)$ given a specific request $x$.

Expert $i$ works on
its assigned request $x$ by searching for information (equivalently, called `researching $x$'), and answers it successfully in that time slot with probability $q_i(x) \doteq \frac{1}{T_i(x)}\leq 1$.
Experts with larger  $q_i(x)$  presumably have deeper knowledge that allows them to quickly research problems, and so, a crude measure of expertise of an expert is $R_i=\sum_x q_i(x)$.
$d_{x,i}(t) =0,1$ indicates failure or success of $i$ finding the answer for $x$ during time slot $t$, respectively.
If the request is not answered
successfully, it goes back in its queue. Future scheduling of that request does not utilize the past history of handling that request.
Thus, the number of (potentially non-consecutive) time slots needed to answer a request is a geometric random variable with average time $T_i(x)$. 
Clearly, the queue lengths update as $Q_{x,i}(t+1)=Q_{x,i}(t)+a_{x,i}(t)-d_{x,i}(t)$. 

The maximum request load $\lambda$ that can be researched by this system, while keeping the
  request queues stable is called {\em capacity}. Queue stability can be defined either as {\em stability-in-the-mean} \cite{kumar}, 
i.e.,
\begin{eqnarray}
\limsup_{T \rightarrow \infty} \frac{1}{T} \sum_{t=1}^{T} \sum_x E[Q_{x,i}(t)] < \infty, \qquad \forall i, \label{eqn:stability}
\end{eqnarray}
or as positive recurrence of the queue Markov chain \cite{bremaud}.
Given the large number of topics (large $\xspace$), we may be willing to
reject requests that do not match the expertise available to research them, i.e., $e_{x,i}(t)\in\{0,1\}$, if a new arriving request
$x$ at expert $i$ is kept or rejected, respectively, at time $t$. So, we will also wish
to characterize capacity under $\varepsilon-$loss constraint. i.e., the maximum load that a system can handle while keeping queues stable, with losses bounded as below.
\begin{eqnarray}
\limsup_{T \rightarrow \infty} \frac{1}{T} \sum_{t=1}^{T} \sum_x E[e_{x,i}(t)] \leq \varepsilon, \qquad  \forall i.  \label{eqn:loss}
\end{eqnarray}

\end{enumerate}

%---------------------------------------------------------------------
\section{Results}
\label{sec:preliminary}

Based on the theoretical framework of information search presented in Section \ref{sec:framework},
we present preliminary results on the performance of the system.

\subsection{Single Expert}
\label{sec:single}

Consider a simple setting with only a single expert `1', as shown in Figure \ref{fig:schedulingproblems}(a). 
At discrete time $t$, requests arrive and are placed in their respective queues. A scheduler assigns a request $x$ from one of the queues
to the expert,
who searches for information to answer it and succeeds in answering it with probability $q(x)$, which depends on the expert's average search time
$T(x)$ for that request. 

\begin{lemma}
\label{thm:single}
The capacity is  $\lambda^* = \left(\sum_x \frac{p(x)}{q(x)}\right)^{-1}$. Further, any $\lambda < \lambda^*$
can be achieved using any work conserving  scheduler (such as one that assigns an arbitrary request in the queue to the expert.) 
\end{lemma}

Evidently, capacity is high if the expertise of the expert  matches closely with the population of requests coming in,
so that none of the ratios $\frac{p(x)}{q(x)}$ is too large. 
In light of this elementary result, we can call $\lambda_j(p) \doteq \left(\sum_x \frac{p(x)}{q_j(x)}\right)^{-1}$ as the capacity
of the expert $j$ with respect to p.m.f. $p(x)$.  

We can also characterize the capacity under loss constraint (\ref{eqn:loss}). 
\begin{lemma}
\label{thm:singleloss}
If we are willing to accept average loss rate $\varepsilon$, the
capacity is no less than the $\lambda^*$ specified by the Linear program below.
\begin{eqnarray}
\lambda^* & = & \max_{\mu(x)} \left( \sum_x \mu(x) \frac{p(x)}{q(x)}\right)^{-1} \qquad \mbox{where}  \label{eqn:caploss} \\
&& \sum_x \mu(x)p(x) \frac{q(x)+\varepsilon}{q(x)} = 1, \label{eqn:caplossequality} \\
&&  0 \leq \mu(x) \leq 1, \forall x. 
\end{eqnarray}
Any $\lambda < \lambda^*$ can be achieved by an offline scheduler; one that first solves this optimization problem assuming known $p(x),q(x)$.
\end{lemma}

The offline scheduler first calculates the probabilities $\mu(x)$  by solving the optimization problem (\ref{eqn:caploss}) before 
considering requests. After that, when request $a_x(t)$ comes in, the scheduler drops it (so $e_{x}(t)=1$) independently with probability $1-\mu(x)$. Otherwise, it gets inserted into its topic queue. 

%The online scheduler is as follows. At time $t$, let the requests awaiting an answer in the queues be  $F(t)=\{x:Q_x(t)>0 \}$.
%If $F(t) \ne \emptyset$, the optimal choice for the expert to research is the request $x^* = \argmax_{x \in F(t)} Q_x(t) \frac{1}{T(x)}$. So, the
%scheduler simply needs to provide all the queue lengths $Q_x(t)$ to the expert, who can then use her intuitive knowledge of her research function $T(x)$
%to select the best request to research. Notice that no knowledge of the request p.m.f. $p(x)$ is needed (which, in any case, cannot be easily
%obtained given the large set $\xspace$.)

For $\varepsilon=0$, the solution to (\ref{eqn:caploss}) is the same as Lemma \ref{thm:single}, because the equality (\ref{eqn:caplossequality}) reduces to $\sum_x \mu(x)p(x) =1$, and so,
 can only be satisfied by $\mu(x) \equiv 1$ (since $p(x)>0$ is a p.m.f.)
Lemma \ref{thm:singleloss} is especially useful when there is a gross mismatch between the requests and the expert. For example, if the expert has $q(x)=0$ iff $x \in \xspace_0$,
the lossless capacity is $\lambda^*=0$. But if we allow loss, we can set $\mu(x)=0, \forall x \in \xspace_0$ and $\mu(x)=1$ otherwise, to achieve a load $\lambda = \left( \sum_{x \notin \xspace_0} \frac{p(x)}{q(x)}\right)^{-1} > 0$,
while accepting a loss of $\varepsilon = \lambda \sum_{x \in \xspace_0} p(x)$.

%With a metric space model of information searching, we can investigate deeper concepts. For example, assume a metric space model with 
%research time $T(x)=1+d(x,S)$ (such as the Graph or Vector space models).  Then, $T(x)$ is Lipschitz-continuous with $\beta=1$. Suppose
%an expert views request $x$, but confuses it for $y$ during the scheduling phase. (Assume that as she researchers the problem, she does eventually
%discover the correct answer $x$.) Requests further away from the expert's zone of knowledge appear fuzzier, and so, we assume the 
%estimation error is $d(x,y) \leq \gamma d(x,S), \ \forall x,y$, for some $\gamma<1$. 

Suppose that the expert has an erroneous estimate $\hat{T}(x)$ of her average searching time $T(x)$. For example, the expert may have an intuitive approximation of these times based on her past
experience answering questions about these topics.
Since the scheduler uses $\hat{T}(x)$ to schedule while the true search time is $T(x)$, the capacity $\lambda^*$ calculated in Lemma \ref{thm:single} may be an over-estimation, resulting in queue instability. However, an achievable load can
be guaranteed if we assume that the estimation error has a known bound, i.e., if we assume $\hat{T}(x) \geq \gamma T(x), \forall x$, for some constant $\gamma\leq 1$.

\begin{corollary}
\label{thm:singleerror}
Let $\lambda^*$ be the capacity in Lemma \ref{thm:single} calculated using the erroneous search times $\hat{T}(x)$ that have bounded errors. Then, any work conserving scheduler using $\hat{T}(x)$ can achieve any load less than  $\lambda < \gamma \lambda^*$ with stable queues. 
\end{corollary}

\begin{figure}[t]
\centering
\begin{subfigure}{.45\textwidth}
  \centering
\includegraphics[height=1.5in]{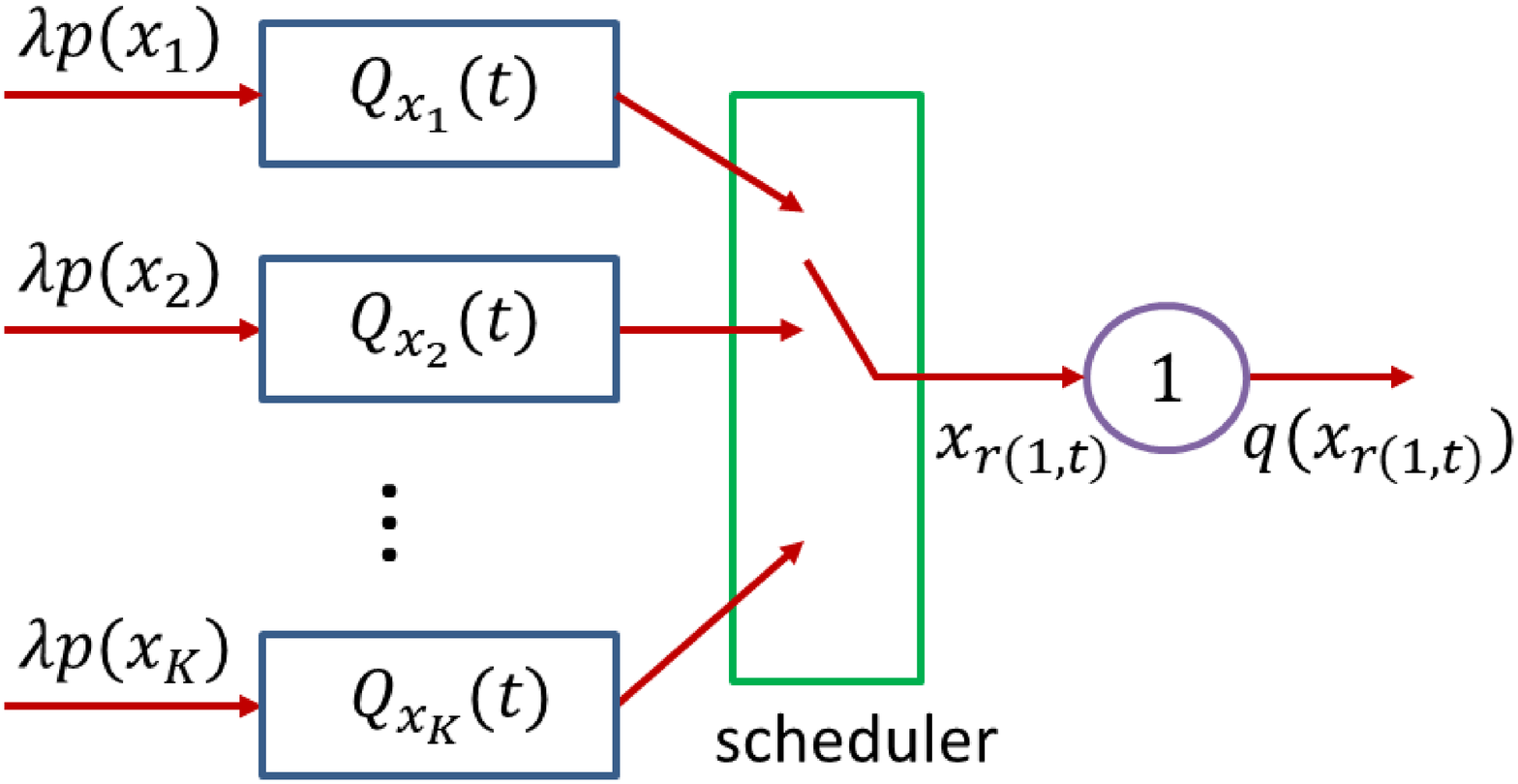}
\subcaption{Single expert}
\label{fig:singlescheduling}
\end{subfigure}%
\begin{subfigure}{.45\textwidth}
  \centering
\includegraphics[height=1.5in]{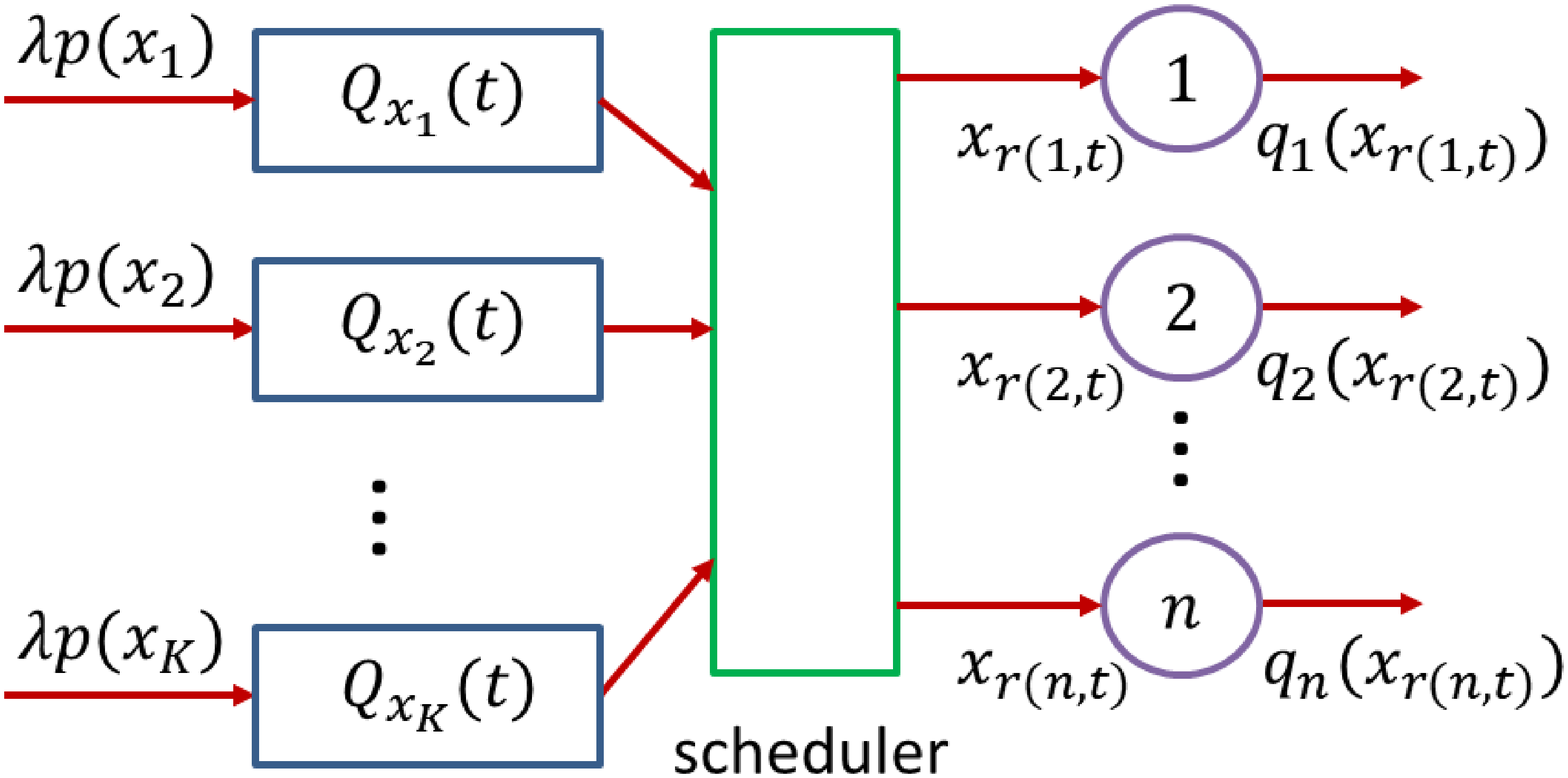}
\subcaption{Multiple coordinating experts in complete graph}
\label{fig:multiplescheduling}
\end{subfigure}%
\caption{Expert scheduling scenarios.}
\label{fig:schedulingproblems}
\end{figure}

\subsection{Multiple Coordinating Experts}
\label{sec:multiple}

Now consider $n$ experts on a social network with a complete graph (so that they can all see requests in each others' queues). Since the theoretical framework allows the scheduler to schedule requests from a neighbor's queue,
in the complete graph case, we can equivalently assume that the queues of all the experts are merged together for each topic $x$; i.e., $Q_x(t)= \sum_i Q_{x,i}(t)$. Define $p(x)=\sum_i p_i(x)$
as the merged p.m.f.
See Figure \ref{fig:multiplescheduling}. This models
experts that each monitor a single knowledge market like {\tt Quora}.
In this case, we have the following result.

\begin{lemma}
\label{thm:multiple}
The capacity with multiple coordinating experts is at least
\begin{eqnarray}
\lambda^* & = & \left(\max_{(\alpha_i)} \sum_x \min_i \left(\alpha_i \frac{p(x)}{q_i(x)}\right) \right)^{-1} \qquad \mbox{where} \label{eqn:multiplecoordinating} \\
&&  \sum_i \alpha_i=1, \quad \alpha_i \geq 0 \nonumber
\end{eqnarray}
Further, any $\lambda < \lambda^*$
can be achieved using an offline scheduler. 
\end{lemma}

  The offline scheduler is assumed to know $p(x),q_i(x)$. It maintains separate topic queues $Q_{x,i}(t)$ for each expert $i$.
Before considering requests, it first calculates the solution to the convex dual problem \cite{boyd} of the maximization problem over $\alpha_i$ stated in (\ref{eqn:multiplecoordinating}).
(For brevity, we will simply call this maximization problem as the problem (\ref{eqn:multiplecoordinating}).) The dual problem is
the Linear program below (see Lemma \ref{thm:dualmultiple}).
\begin{eqnarray}
\min_{\mu,  \ s_{i,x}}  && \hspace{-2ex} \mu  \quad \mbox{s.t.} \label{eqn:dualmultiple} \\
&& \sum_x \frac{p(x)}{q_i(x)}s_{i,x} \leq \mu, \ \forall i \label{eqn:dualmultiplefirst} \\
&&  \sum_i s_{i,x}=1, \ \forall x, \quad s_{i,x} \geq 0, \ \forall i,x \label{eqn:dualmultiplesecond}
\end{eqnarray}
Using these pre-computed $s_{i,x}$ (which we note is a p.m.f. over $i$ for each $x$), for each arriving request $x$, the scheduler selects an expert $i$ randomly and independently  according to the p.m.f.   $s_{i,x}$, 
and then inserts that request into the topic queue $Q_{x,i}$ of expert $i$. In each time slot, the scheduler also assigns a request randomly to expert $i$ from among the requests queued up at that expert's queues $Q_{x,i}$. Expert $i$ is kept idle
if and only if her own queues are all empty. Thus, the expert is work conserving with respect to her own queues.

%The distributed scheduler is described here. At time
%$t$, the scheduler assigns request $x$ to expert $j$ (denoted by $\sigma_{x,j}(t)=1$, else $0$), so as to solve the following assignment  problem.
%\begin{eqnarray}
%\max_\sigma & & \sum_x Q_{x}(t) \sum_j  q_j(x)\sigma_{x,j}(t) \qquad \mbox{where} \label{eqn:multiplecoordinatingscheduler} \\
%&&  \sum_x \sigma_{x,j}(t)\leq 1, \forall j \nonumber \\
%&&  \sum_j \sigma_{x,j}(t)\leq Q_{x}(t), \forall x  \nonumber
%\end{eqnarray}
%However the solution must be distributed, so that experts do not need to explicitly provide their $q_i(x)$ functions. This can be done using the Distributed auction algorithm \cite{xxx}. In our context, this means that a series of prices are announced by the experts, with each expert $i$ comparing the price to its private
%value of $Q_{x}(t)q_i(x)$ to select its own assignment. That paper showed that the algorithm converges to the optimal assignment after a finite
%number of steps. 

As opposed to single expert scheduling, in this case, any one expert mismatched to the request p.m.f. $p(x)$
may not be catastrophic. In fact, the following case shows that a {\em diversity} of experts may be preferable. Suppose there
 are $n$ experts, with each expert $i$ having expertise $R_i \doteq \sum_x q_i(x)=1$. Consider
a toy case where $|\xspace|=n$ and $p(x)=\frac{1}{n}, \forall x$. If the experts are identical, i.e., $q_i(x)=q(x), \forall x$, then the capacity in 
\eqref{eqn:multiplecoordinating} is maximized for $q(x)=\frac{1}{n},  \forall x$ and it is $\lambda^*=1$. Instead, suppose we
have diverse experts $q_i(x)=\indicator{x=x_i}$ (where $\indicator{A} \in \{0,1\}$ is the indicator function of statement $A$), each of which
also has expertise $R_i=1$ as in the case of identical experts. Then, the capacity
 in \eqref{eqn:multiplecoordinating} is increased to $\lambda^*=n$, showing the benefit of diversity.

\section{Conclusions}
This paper set up a theoretical framework to analyze the dynamic process by which requests for information arrive in a social network,
so that either a single expert or a collection of experts can search for the needed information. Preliminary results on queuing and scheduling
analysis in this framework were presented. Future work will look at online and distributed schedulers for the scenarios analyzed in this paper.

%-------------------------------------------------------------------------
%\begin{thebibliography}{90}
%\input{tempref_v1}
%\end{thebibliography}

%\bibliographystyle{IEEEtran}
%\bibliography{ccf2019v1}
 
\appendix

\section{Proofs}

We will use Lyapunov analysis and invoke the well-known Foster-Lyapunov theorem \cite{bremaud}, which we state below for completeness.
\begin{theorem}[Foster-Lyapunov theorem]
\label{thm:foster}
Suppose a Markov chain $Q(t)$ in a countable state space $E$ is irreducible and suppose there exists a function $L:E \rightarrow 	\mathbb{R}$ bounded below as $L \geq 0$. Suppose also
that there is a finite set $F$ and some $\delta>0$ such that,
\begin{eqnarray}
E[L(Q(t+1))|Q(t)] & <& \infty, \ \  \forall i \in F, \\ 
E[L(Q(t+1))|Q(t)] & <& L(Q(t)) - \delta, \ \ \forall i \notin F.
\end{eqnarray}
Then the Markov chain is positive recurrent. 
\end{theorem}
With a slight abuse of notation, the Lyapunov function is often written as $L(t)$.

\begin{proofnumbered}
{\em [Lemma \ref{thm:single}]:}
Let ${\mathbf Q}(t) \doteq 
[Q_x(t)]$ be the vector of topic queue lengths. Assume that $q(x)>0, \forall x$, since otherwise $\lambda^*=0$ and the Lemma is trivially proved .
To show that any load $\lambda < \lambda^*$ is achievable using any work conserving scheduler, consider the Lyapunov function $L(t) = \sum_x \frac{1}{q(x)}Q_x(t)$ for the irreducible Markov chain ${\mathbf Q}(t)$. 
Then, $\Delta L(t)\doteq L(t+1)-L(t)=
\sum_x \frac{1}{q(x)}(a_x(t)-d_x(t))$. So, $E[\Delta L(t) | {\mathbf Q}(t)] = \sum_x \frac{1}{q(x)}E[a_x(t)-d_x(t)| {\mathbf Q}(t)]  = 
\sum_x \frac{1}{q(x)}(\lambda p(x) - q(x)\sigma_x(t)) = \lambda  \sum_x \frac{p(x)}{q(x)} - \sum_x  \sigma_x(t) $, where $\sigma_{x}(t)=1$ if the scheduler assigns a request from topic $x$ to the expert, else $0$. 
This is because, if an expert works on request $x$, it has a probability  $q(x)$ of successfully answering it in that slot.
Let $B = \{{\mathbf Q}(t):  {\mathbf Q}(t)={\mathbf 0} \}$.
For any work conserving scheduler, $\sum_x  \sigma_x(t) =1$ if ${\mathbf Q}(t) \notin B$.  So, for the case ${\mathbf Q}(t) \notin B$,
 $E[L(t+1) | {\mathbf Q}(t)] = L(t) + \lambda  \sum_x \frac{p(x)}{q(x)} - 1 = L(t) -\delta$, where $\delta \doteq 1-\lambda \sum_x \frac{p(x)}{q(x)} >0$ since $\lambda < \lambda^*$. 
Further, for  the
case ${\mathbf Q}(t) \in B$, 
$E[L(t+1) | {\mathbf Q}(t)] \leq \sum_x  \frac{1}{q(x)} E[a_x(t) | {\mathbf Q}(t)] = \sum_x \frac{\lambda p(x)}{q(x)} = c < \infty$ since we assumed $q(x)>0, \forall x$. 
\par
Considering both cases, by Foster-Lyapunov theorem, the irreducible Markov chain ${\mathbf Q}(t)$ is  positive recurrent, which proves stability. 
Alternatively, stability-in-the-mean can be directly obtained by telescoping the $E[\Delta L(t) | {\mathbf Q}(t)]$ terms. 

\noindent $E[L(t)] \leq E[L(0)] + \sum_{\tau=1}^t \left(-\delta 1({\mathbf Q}(\tau) \ne {\mathbf 0}) + c 1({\mathbf Q}(\tau) = {\mathbf 0}) \right) \leq \max(E[L(0)],c) $.
So,

\noindent
$\limsup_{T \rightarrow \infty} \frac{1}{T} \sum_{t=1}^{T} \sum_x E[Q_{x}(t)] \leq (\max_x q(x)) \limsup_{T \rightarrow \infty} \frac{1}{T} \sum_{t=1}^{T} E[L(t)]$

\noindent $ \leq  \max(E[L(0)],c) (\max_x q(x)) < \infty$. Thus, the chosen $\lambda$ also achieves queue stability-in-the-mean.
\par
For the converse, if $\lambda > \lambda^*$,  $E[L(t+1) | {\mathbf Q}(t)] = L(t)  + E[\sum_x \frac{1}{q(x)}(a_x(t) - d_x(t)) | {\mathbf Q}(t)] = L(t)  + 
\sum_x (\lambda \frac{p(x)}{q(x)} - \sigma_{x}(t)) \geq L(t)  + \lambda \sum_x \frac{p(x)}{q(x)} - 1 = L(t)-\delta$, since the expert can only work on one request in each time slot.
However, since $\lambda > \lambda^*$, we now have $\delta<0$.
Telescoping this result, we get $E[L(t)] \geq  E[L(0)] - t\delta$. Letting $q_{min} = \min_x q(x) >0$, we have
$\frac{1}{T} \sum_{t=1}^{T} \sum_x E[Q_x(t)] \geq q_{min} \frac{1}{T} \sum_{t=1}^{T} \sum_x \frac{1}{q(x)} E[Q_x(t)] = q_{min} \frac{1}{T} \sum_{t=1}^{T} E[L(t)] \geq
q_{min} E[L(0)] - \frac{1}{2} q_{min}\delta (T+1) \rightarrow \infty$ as $T \rightarrow \infty$. Thus, the queues are not stable-in-the-mean.
\end{proofnumbered} 

\vspace{1em} 

\begin{proofnumbered}
{\em [Lemma \ref{thm:singleloss}]:}
Note that the stated optimization problem can be re-written as 
\begin{eqnarray}
&& \max_{\mu(x)} \ \lambda \quad \mbox{s.t.} \label{eqn:singlelossmax} \\
&& \lambda \sum_x \mu(x) \frac{p(x)}{q(x)} \leq  1, \label{eqn:firstinequality} \\
&& \lambda \sum_x (1-\mu(x)) p(x) \leq \varepsilon,   \label{eqn:secondinequality} \\
&&  0 \leq \mu(x) \leq 1, \forall x.
\end{eqnarray}
This is because $\lambda$ is maximized when both inequalities (\ref{eqn:firstinequality}),(\ref{eqn:secondinequality}) are equalities. Probability $\mu(x) $ can be shifted from one inequality to the other until both are equalities. Thus, in the optimal
solution, $\sum_x \mu(x) \frac{p(x)}{q(x)}  = \sum_x (1-\mu(x)) \frac{p(x)}{\varepsilon}$. This is equality (\ref{eqn:caplossequality}) stated in the Lemma.
\par
The offline
scheduler, which drops requests randomly, is equivalent to reducing the expected arrival rate at the queue of $x$ to $\lambda p(x)\mu(x)$. So, by Lemma \ref{thm:single}, 
$\lambda = \left(\sum_x \frac{p(x)\mu(x)}{q(x)}\right)^{-1}$ is indeed achievable with stable queues. For the losses, $E[e_{x}(t)] = E[e_{x}(t) a_x(t)] =
\lambda p(x)  (1-\mu(x))$. So, $\limsup_{T \rightarrow \infty} \frac{1}{T} \sum_{t=1}^{T} \sum_x E[e_{x}(t)] = \limsup_{T \rightarrow \infty} \frac{1}{T} \sum_{t=1}^{T} \sum_x \lambda p(x)  (1-\mu(x))
\leq \varepsilon$ due to (\ref{eqn:secondinequality}). Thus, loss is within the acceptable bound.

\end{proofnumbered}

\vspace{1em} 

\begin{proofnumbered}
{\em [Corollary \ref{thm:singleerror}]:}
Here, $\lambda^* = \left(\sum_x \frac{p(x)}{\hat{q}(x)}\right)^{-1}$, where $\hat{q}(x) = \frac{1}{\hat{T}(x)}$, since the erroneous
$\hat{T}(x)$
is used to calculate capacity. Since $\hat{q}(x) \leq \frac{1}{\gamma} q(x)$, $\lambda^* \leq \left(\sum_x \gamma \frac{p(x)}{q(x)}\right)^{-1}$.
Thus, if the load satisfies $\lambda < \gamma \lambda^*$, we also get $\lambda < \left(\sum_x \frac{p(x)}{q(x)}\right)^{-1}$, where the right hand side is the true capacity of the system. Therefore,
by Lemma 1 (scheduling without errors), such $\lambda$ is achievable with stable queues.

\end{proofnumbered} 

\vspace{1em}

\begin{lemma}
\label{thm:dualmultiple}
The problem (\ref{eqn:dualmultiple}) is the convex dual of problem (\ref{eqn:multiplecoordinating}).
\end{lemma} 
\begin{proof}
The problem (\ref{eqn:multiplecoordinating}) can be written as
\begin{eqnarray}
\max_{\alpha_i, \ \beta(x)} \sum_x \beta(x), && \quad \mbox{s.t.} \\
 \beta(x) &\leq & \alpha_i \frac{p(x)}{q_i(x)}, \ \forall i,x,  \label{eqn:dualineq}  \\
 \sum_i \alpha_i&= &1, \quad \alpha_i \geq 0 \nonumber
\end{eqnarray}
With $s_{i,x} \geq 0$ being the dual variables for inequalities (\ref{eqn:dualineq}), the Lagrangian is $J = \sum_x \beta(x) - \sum_i \sum_x s_{i,x}(\beta(x) - \alpha_i \frac{p(x)}{q_i(x)}) =
\sum_x \beta(x)(1-\sum_i s_{i,x}) + \sum_i \alpha_i \sum_x  \frac{p(x)}{q_i(x)}s_{i,x}$.
Maximizing the Lagrangian over $\beta(x)$ shows that it is finite only when the condition $\sum_i s_{i,x}=1, \forall x$ is imposed.  Then, maximizing the Lagrangian over the p.m.f. $\alpha_i$
gives the dual function $\max_{\alpha_i, \beta(x)} J = 0 + \max_i \sum_x  \frac{p(x)}{q_i(x)}s_{i,x}$. Thus, the convex dual problem is
\begin{eqnarray}
\min_{s_{i,x}} \ \max_i && \hspace{-2em} \sum_x  \frac{p(x)}{q_i(x)}s_{i,x},  \quad \mbox{s.t.} \label{eqn:dualmin2} \\
 \sum_i s_{i,x}&= & 1, \ \forall x, \quad s_{i,x} \geq 0, \ \forall i,x 
\end{eqnarray}
The  minimization in (\ref{eqn:dualmin2}) can be re-written as $\min_{\mu, s_{i,x}} \ \mu$, where $\mu \geq \sum_x  \frac{p(x)}{q_i(x)}s_{i,x}, \forall i$. This gives the dual problem specified in (\ref{eqn:dualmultiple}).

\end{proof}

\vspace{1em}

\begin{proofnumbered}
{\em [Lemma \ref{thm:multiple}]:}
Let $s_{i,x}, \mu^*$ be the optimal solution of the dual problem (\ref{eqn:dualmultiple}). Recollect that the offline scheduler uses this optimal $s_{i,x}$ to assign requests to experts' individual queues.
For expert $i$, the arrival of request $a_{x,i}(t)$ into its queue $Q_{x,i}$ is independent of arrival of other requests to its own queues or to other experts' queues, and has a rate of $\lambda p(x)s_{i,x}$ with load $\lambda$. Since the scheduling of expert $i$ only
considers its own queues, its schedule is independent of schedules of other experts. So we can  analyze the queue stability of each expert $i$ separately. By Lemma \ref{thm:single},
expert $i$'s capacity is $\lambda_i^* = \left(\sum_x \frac{p(x)s_{i,x}}{q_i(x)}\right)^{-1} \geq (\mu^*)^{-1}$ by (\ref{eqn:dualmultiplefirst}). By strong duality, the solutions of (\ref{eqn:multiplecoordinating})
and (\ref{eqn:dualmultiple}) are the same, i.e., $\mu^* = (\lambda^*)^{-1}$. So, $\lambda_i^* \geq \lambda^*$ and also $ \lambda^* > \lambda$ by choice of the load. Thus, the load $\lambda$ seen by expert $i$ is indeed
below its capacity $\lambda_i^*$, and so, Lemma \ref{thm:single} guarantees its queue stability.

\end{proofnumbered}

\end{document}